\documentclass[12pt,a4paper]{amsart}
\usepackage[abbrev]{amsrefs}
\usepackage{amssymb}

\theoremstyle{plain}
  \newtheorem{theorem}{Theorem}[section]

  \newtheorem{proposition}[theorem]{Proposition}

\theoremstyle{definition}

\theoremstyle{remark}
  \newtheorem{remark}[theorem]{Remark}
  \newtheorem*{ack}{Acknowledgment}

\numberwithin{equation}{section}
\begin{document}
\title[]{Change the coefficients of conditional entropies\\ in extensivity}
\author[]{Asuka Takatsu$^{\ast\dagger}$}
\address{\scriptsize$\ast$ Department of Mathematical Sciences, Tokyo Metropolitan University, Tokyo {192-0397}, Japan.}
\email{asuka@tmu.ac.jp}
\address{\scriptsize$\dagger$ Mathematical Institute, Tohoku University, Sendai 980-8578, Japan.}

%
\date{\today}
\maketitle
\vspace{-17pt}
\begin{abstract}
The Boltzmann--Gibbs entropy is a functional on the space of probability measures.
When a state space is countable, 
one characterization of the Boltzmann--Gibbs entropy is given by the Shannon--Khinchin axioms,
which consist of continuity, maximality, expandability and extensivity.
Among these four properties,  the extensivity is generalized in various ways.
The extensivity of a functional is interpreted as the property that,
for any  random variables $(X,Y)$ taking finitely many values in $\mathbb{N}$,
the difference between the value of the functional at the joint law of $(X,Y)$ and  that at the law of $X$
coincides with the linear combinations of the values  at the conditional laws of $Y$ given $X=n$ with coefficients given by the probabilities of each event $X=n$.
A generalization of the extensivity obtained  by replacing the coefficients with a power of the probabilities of the events $X=n$ provides a characterization of the Tsallis entropy.

In this paper, we first prove the impossibility to replace the coefficients with a non-power function of  the probabilities of the events $X=n$.
Then we  estimate the difference between the value at the joint law of $(X,Y)$ and that at the law of $X$ for a general functional.
\end{abstract}

\section{Introduction and results}
The notion of entropy is a fundamental ingredient in statistics, information theory and so on.
The origin of entropy lies in thermodynamics, and the negative of the Boltzmann--Gibbs entropy is considered as an internal energy.
For  $n\in \mathbb{N}$, set  
\[
\mathcal{P}_n:=\left\{  p=(p_j)_{j=1}^n \in \mathbb{R}^n \ \Bigg|\ p_j \geq 0 \quad \text{for\ } 1\leq j \leq n,\quad  \sum_{j=1}^n p_j=1 \right\},
\qquad 
\mathcal{P}:=\bigsqcup_{n\in \mathbb{N}}\mathcal{P}_n.
\]
Then the \emph{Boltzmann--Gibbs entropy} $S_1: \mathcal{P} \to \mathbb{R}$ is defined by 
\[
S_1(p)=-\sum_{j=1}^n p_j \log p_j \qquad\text{for}\ p\in \mathcal{P}_n,
\]
where by convention $0 \log 0:=0$.
The following characterization of the Boltzmann--Gibbs entropy is called the \emph{Shannon--Khinchin} axioms. 
%
\begin{theorem}{\rm(\cite{Khinchin}*{Theorem 1})}
The Boltzmann--Gibbs entropy  is a unique functional $S$ on ~$\mathcal{P}$ that satisfies the following four properties:
\begin{itemize}
\setlength{\leftskip}{-8pt} 
\item[{\rm (I)}]{\rm (continuity)}
$S$ is continuous on $\mathcal{P}_n$ for $n\in \mathbb{N}$.
\item[{\rm (II)}]{\rm (maximality)}
Define   $p^{(n)}\in \mathcal{P}_n$ by $p^{(n)}_j\equiv 1/n$.
Then  $S(p) \leq S(p^{(n)})$ for $p\in \mathcal{P}_n$.
\item[{\rm (III)}]{\rm (expandability)}
For $p=(p_j)_{j=1}^{n}\in \mathcal{P}_{n}$, $S((p_1, \cdots, p_n, 0))=S(p)$.
\item[{\rm (IV)}]{\rm (extensivity)}
For $P=(p^i_j)_{1\leq i\leq m, 1\leq j\leq n}\in \mathcal{P}_{mn}$,  if 
\begin{equation}\label{mat}
p_j: =\sum_{i=1}^m p^i_j>0 \quad\text{for\ }1\leq j\leq n,
\end{equation}
then 
\begin{equation}\label{SK}
S(P)=S((p_1, \cdots, p_n))+\sum_{j=1}^n p_j  S\left( \left( \frac{p^1_j}{p_j}, \cdots \frac{p^m_j}{p_j}\right)\right).
\end{equation}
\end{itemize}
Here the uniqueness is up to  a multiplicative positive constant.
\end{theorem}
If a functional $S$ on $\mathcal{P}$ has density and $-S$ can be regarded as an internal energy, 
then the density  should be concave and vanishes at $0$ (for example, see~\cite{Mc}).
In this case,  $S$ satisfies the properties (I)--(III) but may not the property~(IV)
(see  Proposition~\ref{inter}).
The property~(IV) reduces  to additivity for  independent two systems. 
More generally,  the property (IV) is interpreted as  follows.
For any  random variables $X_k$ on a probability space~$(\Omega, \mathbb{P})$ taking values in $\{1, \cdots, k\}$, 
the difference between the value of the functional at the joint law of~$(X_n,X_m)$ and  that at the law of $X_n$
coincides with the linear combinations of the values  at the conditional laws of $X_m$ given $X_n=j$ 
with coefficients given by the probabilities of each event $X_n=j$, 
where  we consider
\[
p^i_j=\mathbb{P}(X_m=i, X_n=j), \quad
p_j=\mathbb{P}(X_n=j), \quad
\frac{p^i_j}{p_j}= \mathbb{P}(X_m=i \ |\ X_n=j).
\]

The property (IV) is generalized in various ways.
For example, see~\cite{Abe, IS, Tempesta, Tsallis} and references therein.
In particular, for $q>0$ with $q\neq 1$, if we replace the coefficients  in \eqref{SK}  with  the $q$th power of the probabilities of the events $X_n=j$,
that is,  \eqref{SK} is modified as
\begin{equation}\label{SKq}
S(P)=S((p_1, \cdots, p_n))+\sum_{j=1}^n p_j^q S\left( \left( \frac{p^1_j}{p_j}, \cdots \frac{p^m_j}{p_j}\right)\right),
\end{equation}
then this property is satisfied by the Tsallis entropy $S_q$ defined by
\[
S_q(p):=\frac{1}{q-1} \left(1-\sum_{j=1}^n p_j^q \right)
\quad
\text{for}\ p\in \mathcal{P}_n.
\]
Note that 
\[
S_q(p) \xrightarrow{q\to 1} S_1(p) \quad \text{for}\ p\in \mathcal{P}.
\]
Suyari~\cite{Suyari} proved that the Tsallis entropy $S_q$ is a unique functional that satisfies \eqref{SKq} in addition to the properties~(I)---(III) up to  a multiplicative positive constant. 

In this paper, we first show the impossibility to replace the coefficients in~\eqref{SK} with a non-power function of  the probabilities of the events $X_n=j$
for all functionals on $\mathcal{P}$ with $C^2$-density whose second derivative is negative.
\begin{theorem}\label{power}
For $s \in C( [0,1]) \cap C^2((0, 1])$ with $s''<0$ on $(0,1]$, 
define a functional $S:\mathcal{P} \to \mathbb{R}$ by 
\[
S(p):=\sum_{j=1}^n s(p_j)\quad \text{for}\ p\in \mathcal{P}_n.
\]
If there exists $f:(0,1]\to \mathbb{R}$ such that  
\begin{align}\label{f}
S(P) =S ((p_1, \cdots, p_n))+
\sum_{j=1}^n  f(p_j)  S\left( \left( \frac{p^1_j}{p_j}, \cdots \frac{p^m_j}{p_j}\right)\right)
\end{align}
holds for  $P\in \mathcal{P}_{mn}$ satisfying \eqref{mat}, 
then there exists $q> 0$  such that $f(r)=r^{q}$ on $(0,1]$
and $S$ coincides with $S_q$ up to an additive constant and a multiplicative positive constant.
\end{theorem}
%
We next estimate the difference between the value at the joint law of $(X_m,X_n)$ and that  at the law of $X_n$ 
for a general functional on $\mathcal{P}$ which can be regarded as the negative of an  internal energy. 
Let $S$ be a functional on $\mathcal{P}$ with density $s$ such that 
$s$ is concave on~$[0,1]$ and $s(0)=0$.
We notice that 
\[
\widetilde{s}(r):=s(r)-s(1)r 
\]
is also concave on $[0,1]$ and $\tilde{s}(0)=0$.
Moreover,  $\tilde{s}(1)=0$.
Since the functional  $\widetilde{S}$  on $\mathcal{P}$  with density $\widetilde{s}$ satisfies
\[
\widetilde{S}(P)-\widetilde{S}((p_1, \cdots, p_n))=S(P)-S((p_1, \cdots, p_n))
\]
for $P\in \mathcal{P}_{mn}$ satisfying \eqref{mat},
to estimate the difference $S(P)-S((p_1, \cdots, p_n))$, 
we can assume that  the density $s$ of $S$  vanishes at $1$ without loss of generality.
In addition to the concavity $s$ on $[0,1]$,
we assume that  $s\in C^2((0, 1])$ with  $s''<0$ on $(0,1]$ and  the finiteness of 
$ \sup\{ {s''(rt)}/{s''(t)} \ | t\in (0,1] \}$ for $r\in (0,1]$.
Then $ \inf \{ {s''(rt)}/{s''(t)} \ | t\in (0,1] \}$ is obviously nonnegative and finite for $r\in (0,1]$. 
\begin{theorem}\label{diff}
For $s \in C( [0,1]) \cap C^2((0, 1])$ with $s(0)=s(1)=0$ and $s''<0$ on $(0,1]$, 
define a functional $S:\mathcal{P} \to \mathbb{R}$ by 
\[
S(p):=\sum_{j=1}^n s(p_j)\quad \text{for}\ p\in \mathcal{P}_n.
\]
Assume the finiteness of $\sup\{ {s''(rt)}/{s''(t)} \ | t\in (0,1] \}$ for $r\in (0,1]$ and 
define $\underline{f},\overline{f}:(0,1]\to \mathbb{R}$ by
\[
\underline{f}(r):=r^{2}\cdot \inf_{t\in (0,1]} \frac{s''(rt)}{s''(t)},\qquad
\overline{f}(r):=r^{2}\cdot \sup_{t\in (0,1]} \frac{s''(rt)}{s''(t)},
\]
respectively.
Then 
\begin{align*}
&
\sum_{j=1}^n \underline{f}(p_j)  S\left( \left( \frac{p^1_j}{p_j}, \cdots \frac{p^m_j}{p_j}\right)\right)
+s'(1)\sum_{j=1}^n \Big(\overline{f}(p_j) -\underline{f}(p_j)\Big)\\
&\leq  S(P)-S((p_1,\cdots, p_j)) \\
&\leq \sum_{j=1}^n \overline{f}(p_j)  S\left( \left( \frac{p^1_j}{p_j}, \cdots \frac{p^m_j}{p_j}\right)\right)
-s'(1)\sum_{j=1}^n \Big(\overline{f}(p_j) -\underline{f}(p_j)\Big)
\end{align*}
for  $P\in \mathcal{P}_{mn}$ satisfying \eqref{mat}.
\end{theorem}
\begin{remark}
(1)
If $S=S_q$ with $q>0$, 
then  $\underline{f}(r)=\overline{f}(r)=r^q$  on $(0,1]$
and all inequalities  in Theorem~\ref{diff} become equality.

(2)
For $s\in C([0,1]) \cap C^2((0, 1])$ with $s''<0$ on $(0,1]$, 
$ \sup\{ {s''(rt)}/{s''(t)} \ | t\in (0,1] \}=\infty$ may happen for some $r\in (0,1)$.
Indeed, if we define 
\[
s(r):=-\int_0^r \int_0^t u \left(| \cos (1/u)|+ u \right)du dt \quad \text{for\ }r\in [0,1],
\]
then $s\in C([0,1]) \cap C^2((0, 1])$ with
\[
s''(r)=- r \left(| \cos (1/r)|+ r \right)<0\quad \text{for\ }r\in (0,1].
\]
For $t_k:=\{(k+\frac12)\pi\}^{-1}$ with $k\in \mathbb{N}$, 
we observe that 
\[
 \frac{s''(t_k/2)}{s''(t_k)}
 =\frac{\frac12\left(|\cos ( (2k+1) \pi)|+ \tfrac{1}{(2k+1)\pi}\right)}{| \cos ((k+\frac12)\pi)|+ \tfrac{1}{(k+\frac12)\pi}}
 =\frac12\left\{\left(k+\frac{1}{2}\right)\pi+  \frac{1}{2}\right\}\xrightarrow{k\to \infty}\infty,
 \]
implying $ \sup\{ {s''(t/2)}/{s''(t)} \ | t\in (0,1] \}=\infty$.

(3)
Naudts~\cite{Naudts} discussed  a generalization of the Boltzmann--Gibbs entropy $S_1$
via a  continuous, nondecreasing  function $\phi:(0, \infty)\to (0,\infty)$.
This generalized entropy is called the \emph{$\phi$-deformed entropy}  and denoted by $S_\phi$.
For a  functional $S$ on $\mathcal{P}$ with density $s$,
if $s \in C( [0,1]) \cap C^2((0, 1])$,  $s''$ is nondecreasing and $s''(1)<0$, then 
$S$ coincides with $S_{-1/s''}$ on~$\mathcal{P}$ up to an additive constant and a multiplicative positive constant.
%


(4)
Ohta and the author~\cite{OT2} classified $\phi$-deformed entropies in terms of  a quantity 
\[
\theta_\phi:=\sup\left\{  \frac{r}{\phi(r)} \cdot  \limsup_{\varepsilon \downarrow 0} \frac{\phi(r+\varepsilon)-\phi(r)}{\varepsilon} \biggm| r>0\right\} \in [0,\infty],
\]
and analyzed  the Wasserstein gradient flow of $-S_\phi$,
where  $\theta_\phi<2$ is assumed.
Note that the Wasserstein gradient flow of $-S_\phi$ is  the evolution equation associated  to the internal energy $-S_\phi$.
When $\phi=-1/s''$, 
the assumption~$\theta_\phi<2$ guarantees the finiteness of $ \sup\{ {s''(rt)}/{s''(t)} \ | t\in (0,1] \}$ for $r\in(0,1]$.

(5)
If $s$ is  concave on $[0,1]$ with $s(0)=0$, then 
\[
s(\lambda r) \geq  (1-\lambda )s(0)+\lambda s(r)=\lambda  s(r) \quad 
\text{for}\ r\in(0,1), \quad \lambda \in [0,1].
\]
This observation implies 
\begin{align*}
S(P)
=\sum_{j=1}^n \sum_{i=1}^m s(p^i_j)
=\sum_{j=1}^n \sum_{i=1}^m s\left(\frac{p^i_j}{p_j} p_j \right)
\geq
\sum_{j=1}^n \sum_{i=1}^m \frac{p^i_j}{p_j} s(p_j)
=
\sum_{j=1}^n  s(p_j)
=
S((p_1, \cdots, p_n))
\end{align*}
for $P\in \mathcal{P}_{mn}$ satisfying \eqref{mat}.
Thus the first inequality in Theorem~\ref{diff} is trivial unless
\begin{align}\label{iff}
\sum_{j=1}^n \underline{f}(p_j)  S\left( \left( \frac{p^1_j}{p_j}, \cdots \frac{p^m_j}{p_j}\right)\right)
+s'(1)\sum_{j=1}^n \Big(\overline{f}(p_j) -\underline{f}(p_j)\Big)\geq  0.
\end{align}
However, the inequality \eqref{iff} may fail even for a $\phi$-deformed entropy $S_\phi$ with $\theta_\phi<2$.
Indeed, if we define
\[
s(r):=-\int_0^r \log \sin \left(\frac{\pi}4 t \right) dt +Cr \quad \text{for}\ r\in [0,1], 
\qquad
\text{where}\quad  C:=\int_0^1 \log \sin \left(\frac{\pi}4 t \right) dt ,
\]
then $s\in C([0,1]) \cap C^2((0,1])$ with $s(0)=s(1)=0$ and  
\[
s''(r):=-\frac{\pi}{4}\cot \left(\frac{\pi}4 r \right) < 0\quad \text{for}\ r\in (0,1].
\]
This leads to $s'(1)<0.$
Define $\phi:(0, \infty) \to (0,\infty)$ by 
\[
\phi(r):=\begin{cases} 
\dfrac{4}{\pi} \tan \left(\dfrac{\pi}4 r \right)\ & r \in (0,1],  \\ 
2(r-1)+\dfrac{4}{\pi}&  \quad \text{for}\ r>1,
\end{cases}
\]
which is continuous, nondecreasing  on $(0,\infty)$ with $\theta_\phi=\pi/2<2$ and $s''=-1/\phi$ on $(0,1]$.
We find that 
\[
2 =\inf_{t\in (0,1]} \frac{s''(t/2)}{s''(t)} 
\leq \frac{s''(u/2)}{s''(u)}
=\frac{\cos(\pi u/4)+1}{\cos(\pi u/4)} 
\leq  \sup_{t\in (0,1]} \frac{s''(t/2)}{s''(t)}=1+\sqrt{2}
\quad
\text{for}\ u\in (0,1].
\]
Then, for $P\in \mathcal{P}_{22}$ defined by
\begin{equation}\label{x}
p^1_1=\frac12, \quad p^2_1=0, \quad
p^1_2=\frac12x, \quad p^2_2=\frac12(1-x), \text{\ with\ } x\in (0,1),
\end{equation}
it turns out that 
\begin{align*}
&\sum_{j=1}^2 \underline{f}(p_j)  S\left( \left( \frac{p^1_j}{p_j}, \frac{p^2_j}{p_j}\right)\right)
+s'(1)\sum_{j=1}^2 \Big(\overline{f}(p_j) -\underline{f}(p_j)\Big) \\
&= \frac{2}{4} (s(x)+s(1-x))+2s'(1) \cdot\frac{1}{4}\left(1+\sqrt{2}-2\right)
\xrightarrow{x\downarrow 0} \frac{s'(1)}{2}(\sqrt{2}-1)<0.
\end{align*}
Thus the inequality \eqref{iff} fails  for $P$  given by \eqref{x} if $x\in(0,1)$ is small enough.
\end{remark}

The rest of this  paper is organized as follows.
We first confirm  that the negative of an internal energy should satisfy the properties (I)--(III) (Proposition~\ref{inter}).
Then we prove  Theorem~\ref{power} and  Theorem~\ref{diff}.

\section{Proofs}
For a function  $s:[0,\infty)\to \mathbb{R}$, 
define a functional $S$ on $\mathcal{P}$ by
\[
S(p):=\sum_{j=1}^n s(p_j) \quad \text{for}\ p\in \mathcal{P}_n.
\] 
This functional can be extended to a functional on the space of probability measures on a continuous state space.
If we regard $-S$ as an  internal energy, then it is natural to assume $s(0)=0$  since the energy of no matter  should be zero.
To be  physical,  the  pressure function of $-S$ should be nonnegative and nondecreasing,  
in which case $s$ is concave  (for example, see~\cite{Mc}).
Under theses conditions, $S$ satisfies the properties (I)---(III).
Note that 
if we restrict our attention to a countable state space, then it suffices that $s$ is defined on~$[0,1]$, not on~$[0,\infty)$.
\begin{proposition}\label{inter}
For a concave  function  $s:[0,1]\to \mathbb{R}$ with $s(0)=0$, 
the functional $S$ on~$\mathcal{P} $ with density $s$ satisfies the three properties {\rm(I),(II)} and {\rm (III)}.
\end{proposition}
\begin{proof}
The concavity of $s$ on $[0,1]$ leads to the continuity of  $s$ on $[0,1]$.
Consequently, $S$  is continuous on $\mathcal{P}_n$ and  (I) follows.

We next confirm (II).
For  $p=(p_j)_{j=1}^n\in \mathcal{P}_n$,
it follows from the concavity of  $s$ that 
\[
S(p)
=\sum_{j=1}^n s(p_j)
=n \sum_{j=1}^n \frac1n s(p_j)
\leq n s\left(\sum_{j=1}^n \frac1n  p_j  \right)
=ns(1/n )
=\sum_{j=1}^n s(1/n )= S(p^{(n)}).
\]

As for (III),  we deduce  from $s(0)=0$ that
\[ 
S((p_1, \cdots, p_n, 0))=S(p)
\quad {for\ } p=(p_j)_{j=1}^{n}\in \mathcal{P}_{n}.
\]
This completes the proof of the proposition.
\end{proof}
%
By Aleksandrov's theorem,  
a concave function defined on an interval is twice differentiable almost everywhere and its second derivative is nonpositive.
In Theorems~\ref{power} and \ref{diff}, 
the density of a functional on $\mathcal{P}$ is assumed to satisfy 
a slightly stronger condition than concavity,
that is, the twice continuous differentiability and the negativity of the second derivative.
\begin{proof}[Proof of  Theorem~{\rm\ref{power}}]
Let $S$ be a functional on $\mathcal{P}$ with density $s\in C( [0,1]) \cap C^2((0, 1])$ such that  $s''<0$ on $(0,1]$.
Assume that there exists $f:(0,1]\to \mathbb{R}$ satisfying \eqref{f}.
We find that $f(1)=1$.
Fix $r\in (0,1)$ and $\xi\in (0,1]$.
For $x\in (0,\xi)$,  we define $P\in \mathcal{P}_{32}$ by
\begin{align*}
&p^1_1=rx, \quad p^2_1=r(\xi-x),\quad p^3_1=r(1-\xi),\quad
p^i_2=\frac{1-r}{3}\quad \text{for\ }1\leq i\leq 3.
\end{align*}
We substitute this matrix $P$ into \eqref{f} and differentiate it   twice with respect to $x$ to have
\begin{align}\label{twice}
r^2\left\{ s''(rx)+s''(r(\xi-x))\right\}=f(r) \left\{s''(x)+s''(\xi-x)\right\}.
\end{align}
The choice of $x=\xi/2$ implies
\begin{align}\label{ff}
\frac{f(r)}{r^2}=\frac{s''(rx)}{s''(x)} \quad \text{for\ } x\in(0,1/2] \quad r\in (0,1).
\end{align}
Hence  $f$ is continuous on $(0,1]$.
In \eqref{twice}, if we take  $x=t\xi$ with  $\xi \in (0,1/2]$ and $t\in [1/2,1)$,  then  $t\xi \in (0,1/2), (1-t)/t \in(0,1]$ and
\begin{align*}
\frac{f(r)}{r^2}
=
\frac{s''(rt\xi)+s''(r(1-t)\xi)}{s''(t\xi)+s''((1-t)\xi)}
&=
\left(\frac{s''(rt\xi)}{s''(t\xi)}+\frac{s''(r\frac{(1-t)}{t} t\xi)}{s''(t\xi)} \right)\bigg/
\left(\frac{s''(t\xi)}{s''(t\xi)}+\frac{s''(\frac{1-t}{t}t\xi)}{{s''(t\xi)}}\right)\\
&=
\left( \frac{f(r)}{r^2} +\frac{f(r\frac{(1-t)}{t})}{(r\frac{(1-t)}{t})^2} \right)\bigg/
 \left(1+\frac{f(\frac{1-t}{t})}{(\frac{1-t}{t})^2}\right)
\end{align*}
for $r\in (0,1)$, which means that 
\[
f(r)f(\tfrac{1-t}{t})
=f(r\tfrac{(1-t)}{t})\quad
\text{for\ } r\in (0,1], \quad t\in [1/2, 1).
\]
By the continuity of $f$ on $(0,1]$,  there exists $q\in \mathbb{R}$ such that $f(r)=r^{q}$ on $(0,1]$.
We deduce from~\eqref{ff} that 
\[
s''(rx)=s''(x) r^{q-2} \quad\text{for\ } x\in(0,1/2], \quad r\in (0,1].
\]
Substituting it into \eqref{twice} and choosing $ x\in (0, 1/2], \xi=1,$ yield that 
\[
s''(r(1-x))=s''(1-x) r^{q-2}\quad\text{for\ } x\in(0,1/2], \quad r\in (0,1].
\]
Hence  $s''(r)=s''(1)r^{q-2}$ holds on $(0,1]$.
Note that 
\[
s'(1) (1-r)-s(1)+s(r)
=\int_r^1\left(  s'(1)-s'(t) \right) dt= \int_r^1 \int_t^1 s''(u) du  dt
\quad \text{for\ }r \in (0,1].
\]
If $q=1$, then 
\[
s'(1)(1-r) -s(1)+s(r) = s''(1) \left(1 +r\log r-r  \right),
\]
that is, 
\begin{align*}
s(r)=&k_1 r\log r+a_1r +b_1 \quad \text{for\ }r \in [0,1], \\
&\text{where\ } k_1:=s''(1)<0, \quad a_1=-s''(1)+s'(1),\quad b_1:=s''(1)-s'(1)+s(1).
\end{align*}
Similarly, for $q=0$, it turns out that 
\[
s'(1) (1-r)-s(1)+s(r)
=- s''(1)\left(1-r+ \log r\right)\quad
\text{for\ }r \in (0,1].
\]
When $r \downarrow 0$,
the right-hand side diverges  while the left-hand side  converges.
Thus $q=0$ is inappropriate.
Finally, if $q\neq 0,1$, then 
\[
s'(1) (1-r)-s(1)+s(r)
= \frac{s''(1)}{q-1}\left(1-r-\frac{1-r^{q}}{q}\right)
\quad \text{for\ }r \in (0,1].
\]
Since  the limit of the left-hand side exists as $r \downarrow 0$, we find $q>0$ and 
\begin{align*}
s(r)=&k_q\frac{r^{q}-r}{q-1}+a_qr+b_q,  \quad \text{for\ }r \in (0,1],\\
&\text{where\ } k_q:=\frac{s''(1)}{q}<0, \quad a_q:=-\frac{s''(1)}{q}+s'(1), \quad b_q:=\frac{s''(1)}{q}-s'(1)+s(1).
\end{align*}
Thus there exists $q>0$ such that $f(r)=r^{q}$  on $(0,1]$ and 
\[
S(p)=\sum_{j=1}^n s(p_j)
=-k_q S_q(p)+a_q \sum_{j=1}^n p_j+b_q
=-k_q S_q(p)+s(1)
\quad \text{for}\ p\in \mathcal{P}_n.
\]
This proves  the theorem.
\end{proof}
%
\begin{proof}[Proof of Theorem~\rm{\ref{diff} }]
Let  $s \in C( [0,1]) \cap C^2((0, 1])$ such that $s(0)=s(1)=0$ and $s''<0$ on $(0,1]$.
Fix $P\in \mathcal{P}_{mn}$ satisfying \eqref{mat} and $1\leq j \leq n$.
Given $u\in (0,1)$, we deduce from the definition of $\overline{f},\underline{f}$ together with the negativity of $s''$ that 
\begin{align*}
\underline{f}(p_j)s''(u)\geq  p_j^2 s''(p_j u) \geq \overline{f}(p_j)s''(u). 
\end{align*}
For $a,r \in(0,1]$, 
since 
\[
\int_0^r \int_ t^1 a^2s''(au) du dt=a s'(a) r-s(a r),
\]
integrating the above inequalities on  $u \in [t,1]$ and then $t\in [0,r]$ gives 
\begin{align}\label{once}
\underline{f}(p_j) \left( s'(1)r-s(r) \right)
\geq p_j s'(p_j)r-s( p_j r) 
\geq\overline{f}(p_j) \left( s'(1)r-s(r)  \right).
\end{align}
Choosing $r=p^i_j/p_j$ and summing over $1\leq i \leq m$,  we find that 
\begin{align*}
 \underline{f}(p_j) s'(1)  -\underline{f}(p_j) \sum_{i=1}^m s(p^i_j/p_j )
\geq p_j s'(p_j)- \sum_{i=1}^m s(p^i_j) 
\geq \overline{f}(p_j)  s'(1)-\overline{f}(p_j) \sum_{i=1}^m s(p^i_j/p_j ).
\end{align*}
The sum of these inequalities over $1\leq j\leq n $  gives
\begin{align*}
&
\sum_{j=1}^n \underline{f}(p_j) s'(1) 
-
\sum_{j=1}^n \underline{f}(p_j)  S\left( \left( \frac{p^1_j}{p_j}, \cdots \frac{p^m_j}{p_j}\right)\right)\\
&\geq   \sum_{j=1}^n p_j s'(p_j)-S(P) \\
&\geq \sum_{j=1}^n \overline{f}(p_j) s'(1)-\sum_{j=1}^n \overline{f}(p_j)  S\left( \left( \frac{p^1_j}{p_j}, \cdots \frac{p^m_j}{p_j}\right)\right),
\end{align*}
implying
\begin{align*}
&
\sum_{j=1}^n \underline{f}(p_j)  S\left( \left( \frac{p^1_j}{p_j}, \cdots \frac{p^m_j}{p_j}\right)\right)
+\sum_{j=1}^n \Big( -\underline{f}(p_j) s'(1) +p_j s'(p_j) -s(p_j) \Big)\\
&\leq  S(P)-S(p_1,\cdots, p_j) \\
&\leq \sum_{j=1}^n \overline{f}(p_j)  S\left( \left( \frac{p^1_j}{p_j}, \cdots \frac{p^m_j}{p_j}\right)\right)
+\sum_{j=1}^n \Big(- \overline{f}(p_j) s'(1) +p_j s'(p_j) -s(p_j) \Big).
\end{align*}
In \eqref{once}, if we choose $r=1$, then 
\[
\underline{f}(p_j)  s'(1) \geq p_j s'(p_j)-s( p_j )  \geq  \overline{f}(p_j) s'(1).
\]
Substituting these into the above inequalities completes the proof of the theorem.
\end{proof}

\begin{ack}
The author expresses her sincere gratitude to Professor Hiroshi Matsuzoe for stimulating discussion. 
This work was supported in part by JSPS Grant-in-Aid for Scientific Research (KAKENHI) 19K03494.
\end{ack}

\end{document}